\documentstyle[11pt]{article}

\newtheorem{theorem}{Theorem}[section]

\newtheorem{lemma}[theorem]{Lemma}

\begin{document}

\title{Complexity of Acyclic Colorings of Graphs and Digraphs with Degree and Girth Constraints}

\author{
Tom\'{a}s Feder \\
   268 Waverley St., Palo Alto, CA 94301, USA \\
   {\tt tomas@theory.stanford.edu}\\
{\tt http://theory.stanford.edu/${}_{\bf\tilde{}}{\hspace{.02in}}$tomas},
\\
and \\
Pavol Hell \\
   School of Computing Science \\
   Simon Fraser University \\
   Burnaby, B.C., Canada V5A 1S6 \\
   {\tt pavol@cs.sfu.ca}
\\
and \\
Carlos Subi \\
   {\tt carlos.subi@hotmail.com}
}

\date{}


\newcommand{\DEF}[1]{{\em #1\/}}

\newcommand\chic{\chi_c}
\newcommand\C{\hbox{${\cal C}$}}
\newcommand{\RR}{\mbox{$\mathbb R$}}
\newcommand{\NN}{\mbox{$\mathbb N$}}
\newcommand{\ZZ}{\mbox{$\mathbb Z$}}
\newcommand{\eopf}{\raisebox{0.8ex}{\framebox{}}}
\newcommand{\dist}{\hbox{\rm d}}
\renewcommand\a{\alpha}
\renewcommand\b{\beta}
\renewcommand\c{\gamma}
\renewcommand\d{\delta}
\newcommand\D{\Delta}
\newcommand{\directedchi}{\mbox{$\vec{\chi}$}}
\newcommand{\directedE}{\mbox{$\vec{E}$}}
\newcommand{\directedG}{\mbox{$\vec{G}$}}
\newcommand{\directedK}{\mbox{$\vec{K}$}}

\newenvironment{proof}%
{\noindent{\bf Proof.}\ }%
{\hfill\eopf\par\bigskip}%


\maketitle

\begin{abstract}
We consider acyclic $r$-colorings in graphs and digraphs: they color the
vertices in $r$ colors, each of which induces an acyclic graph or digraph. (This
includes the dichromatic number of a digraph, and the arboricity of a graph.)
For any girth and sufficiently high degree, we prove the NP-completeness of
acyclic $r$-colorings; our method also implies the known analogue for classical
colorings. The proofs use high girth graphs with high arboricity and
dichromatic numbers. High girth graphs and digraphs with high chromatic and
dichromatic numbers have been well studied; we re-derive the results from a
general result about relational systems, which also implies the similar fact
about high girth and high arboricity used in the proofs. These facts concern
graphs and digraphs of high girth and low degree; we contrast them with 
acyclic colorings of tournaments (which have low girth and high degree). 
We prove that even though acyclic two-colorability of tournaments is known 
to be NP-complete, random acyclically $r$-colorable tournaments allow
recovering an acyclic $r$-coloring in deterministic linear time, with high probability.
\end{abstract}

\section{Introduction}
Let $G$ be either a graph or a digraph. An {\em acyclic $r$-coloring} of $G$ is an assignment of 
$r$ colors to the vertices in $G$, so that the vertices of each color $i$ induce an acyclic subgraph 
$G_i$ of $G$. Note that the vertex sets $V(G_i)$ partition $V(G)$.  An equivalent condition on the 
$r$-coloring is that no cycle $C$ in $G$ is monochromatic. (In the case of digraphs $C$ is a directed 
cycle.) The least $r\geq 1$ such that $G$ admits an acyclic $r$-coloring is called the {\em arborocity} 
of $G$ if $G$ is a graph~\cite{arb}, and is called the {\em dichromatic number} of $G$ if $G$ is a 
digraph~\cite{b}. In both cases, the problem of deciding if $G$ has an acyclic $r$-coloring is NP-complete, 
for any fixed $r\geq 2$ \cite{hardness}. In this paper we study the effect of restrictions on the girth
and the degrees in $G$. For digraphs, we define the {\em directed girth} to be the minimum length of a
directed cycle, if one exists, and leave it undefined otherwise; and we take the degree of a vertex to be 
the sum of its in-degree and out-degree.

We focus on a combination of girth and degree constraints, and we look at the two opposite ends of the 
spectrum: small girth with high degree on the one hand, and large girth with small degree on the other.
For the former problem, graphs of high degree and small girth are typified by complete graphs and the 
arboricity of complete graphs is trivial, but digraphs of high degree and small directed girth are typified 
by tournaments, and the dichromatic number of tournaments is already a hard problem: even just deciding 
acyclic two-colorability of a tournament is known to be NP-complete \cite{bor2} (see also \cite{bor}). We 
prove that for random acyclically $r$-colorable tournaments $T$ we can recover the unknown acyclic 
$r$-coloring in deterministic linear time, with high probability over the choices of $T$. (Such a coloring 
is unique with high probability, as long as $r$ is a constant.) This underscores the fact that the 
NP-completeness does not come from random instances. We placed this discussion in the last 
section, as it is quite technical.

For the latter problem, we consider graphs and digraphs of low degree and high girth. It is known 
that in the context of classical graph colorings, for each $r$ and $k$ there exists a $d$ such that 
deciding $r$-colorability of graphs with girth at least $k$ and all degrees at most $d$ is NP-complete 
\cite{hougardy}). We offer a simple proof of this fact to illustrate our techniques, and then prove an 
analogous result for acyclic colorings. For both graphs and digraphs, we consider the special cases 
of $r=2$ separately, as we can offer simpler proofs and/or better bounds. In any case, even at this 
opposite end of the scale, the arboricity and the dichromatic number remain mostly NP-complete.

Our NP-completeness constructions depend on gadgets constructed from graphs and digraphs with
high girth and high arboricity or dichromatic number. There are well known constructions of graphs
and digraphs with high girth and high chromatic and dichromatic numbers \cite{erdos,bor2}. As far as we were
able to determine, there does not appear to be such a result for arboricity, so we provide a short proof.
The proof depends on a general result for high girth relational systems from \cite{fv}; in fact the same
result implies the corresponding results for the chromatic and dichromatic numbers as well.

To highlight the gap for digraphs between the largest acyclic induced subgraph and the largest acyclic 
induced subgraph that can be found in polynomial time, we prove that even for digraphs (without digons) 
that have an acyclic $n^{\epsilon}$-coloring, and hence must have acyclic subgraphs of size $n^{1-\epsilon}$, 
it is hard to find one of size greater than $n^{1/2 + \epsilon}$ (Theorem \ref{2.5}).

\section{Graphs and Digraphs with High Girth and Low Degrees}

The following prototype result for ordinary $r$-coloring of undirected graphs is known \cite{hougardy}. 
We offer an easy proof to illustrate our technique.

\begin{theorem}
There exists a function $d=d(r,k)$ such that given $r,k\geq 3$,
the $r$-colorability of a graph $G$ of girth at least $k$, and of 
maximum degree at most $d=d(r,k)$, is NP-complete. 
\end{theorem}

\begin{proof}
We shall reduce from the problem of $r$-colorability. Given a graph $G$, we shall
construct a graph $G'$ with maximum degree at most $d$ and girth at least $k$
that is $r$-colorable if and only if $G$ is $r$-colorable. The first step is to replace
each vertex $x$ with a binary tree $T_x$ having $\deg(x)$ leaves, and place each 
edge $xy$ of $G$ between the $y$-th leaf of $T_x$ and the $x$-th leaf of $T_y$.
The resulting graph $G^*$ has maximum degree three. In the second step we shall
replace each edge of every tree $T_x$ by a gadget $J$ designed to ensure that the girth
of the resulting graph $G'$ is at least $k$ and that all vertices that were leaves of any
one $T_x$ obtain the same colour in any $r$-colouring of $G'$. This ensures that $G'$
is $r$-colourable if and only if $G$ is $r$-colourable. The maximum degree of $G'$ is
then three times the maximum degree $\Delta(J)$ of the gadget $J$, so $d(r,k) = 3 \Delta(J)$.
It remains to construct $J$: it is well known that for any $k$ and $r$ there exists a graph $K$
which is not $r$-colorable and has girth at least $k$ \cite{erdos} (for a constructive proof
see \cite{something}). We may assume that $K$ contains an edge $uv$ such that $K - uv$
is $r$-colorable. We let $J$ to be the graph $K - uv$, and replace each edge $st$ of every 
$T_x$ by a copy of $J$, identifying $s$ with $u$ and $t$ with $v$. Since the girth of $K$ was 
at least $k$, each path joining $u$ and $v$ in $J$ has at least $k$ vertices, and the girth of
the entire graph $G'$ is at least $k$. Any $r$-colouring of $G'$ assigned the same color to
$u$ and $v$ in each copy of $J$, since otherwise $K$ would have been $r$-colorable.
\end{proof} 

The same technique can be applied to acyclic coloring problems. We start with discussing
the computational complexity of arboricity of graphs of high girth and low degree, and prove 
that it remains NP-complete. We begin with the special case of $r=2$.
 
\begin{theorem}\label{3.2}
There is a function $d(k)$ such that given $k\geq 3$,
the problem of acyclic two-coloring a graph $G$ of girth at least
$k$, and of maximum degree at most $d(k)$, is NP-complete. 
\end{theorem}

\begin{proof}
Fix $k\geq 3$. In this case we reduce from the not-all-equal $k$-satisfiability 
problem with three occurences of each variable, which is NP-complete by 
Feder and Ford~\cite{ff}. An instance of this problem is a set of variables with
binary values, and a set of clauses each consisting of $k$ variables. The
question is if values can be chosen so that no clause has all variables of the
same value. (This is also known as the two-colorability problem for $k$-uniform
hypergraphs \cite{hardness}.) Given such an instance, we replace each clause 
by a disjoint cycle of $k$ vertices, one corresponding to each variable. Clearly, 
in every acyclic two-colouring each of these cycles will receive both colours.
We need to ensure that all three occurences of a variable are given the same
value. We shall add for every variable $x$ a simple claw $T_x=K_{1,3}$ with 
the three leaves identified with the three occurences of $x$. The resulting 
graph $G'$ has maximum degree three. For this construction we shall similarly 
use a gadget $J$ to replace each edge of every claw $T_x$; the gadget $J$ will 
ensure that the final graph $G$ has girth at least $k$ and have the same color on 
the three occurences of each variable, in every acyclic two-colouring of $G$.
This guarantees that $G$ is acyclically two-colorable if and only if the original
instance is satisfiable. We construct $J$ from a graph $K$ that has girth at least $k$
that does not admit an acyclic $r$-coloring, but contains an edge $uv$ such that 
$J = K - uv$ is acyclically $r$-colorable. We then replace each edge of each $T_x$
by a copy of $J$, identifying $u$ with $x$ and $v$ with $y$. The construction of 
such a graph $K$ is discussed in the next section.
\end{proof}

The general result is the following.

\begin{theorem}\label{3.1}
There exists a function $d=d(r,k)$ such that given $r,k\geq 3$,
the problem of acyclic $r$-coloring a graph $G$ of girth at least  
$k$, and maximum degree at most $d=d(r,k)$, is NP-complete. 
\end{theorem}

\begin{proof}
Here we combine both of the previous tricks. We again reduce from 
the problem of graph $r$-colorability. Thus let $G$ be an instance. We again replace 
each vertex $x$ by a binary tree $T_x$ with $deg(v)$ leaves. We will use two 
gadgets, $J_1$ and $J_2$ with the following properties:

\begin{itemize}
\item
$J_1$ has vertices $u_1, v_1$ such that any acyclic $r$-coloring of $J_1$ assigns 
the same color to $u_1$ and $v_1$;
\item
$J_2$ has vertices $u_2, v_2$ such that any acyclic $r$-coloring of $J_2$ assigns 
different colors to $u_2$ and $v_2$; 
\item
$J_1$ has girth at least $k$ and contains no path with fewer than $k$ vertices
between $u_1$ and $v_1$; and
\item
$J_2$ has girth at least $k$.
\end{itemize}

Each edge $st$ of each $T_x$ will be replaced by $J_1$ identifying $s$ with $u_1$ and 
$t$ with $v_1$, and each edge $xy$ of $G$ will similarly be replaced by a copy of $J_2$ 
between the corresponding leaves of $T_x$ and $T_y$. The resulting graph $G'$ has girth 
at least $k$ because every cycle in $G'$ is either inside a copy of some $J_i$, or passes 
through some $J_1$. Clearly, $G'$ is acyclically $r$-colorable if and only if $G$ is $r$-colorable 
in the usual sense. The degrees of $G'$ are maximized by three times the maximum degree of 
any $u_1, u_2, v_1, v_2$ in $J_1, J_2$.

It remains to explain how to construct $J_1, J_2$. Let $K$ be a graph of girth $k$
that is not acyclically $r$-colorable but has an edge $uv$ such that $K - uv$ is
acyclically $r$-colorable. (These graphs are constructed in the next section.) Let
$J_1 = K - uv$ and let $J_2$ be obtained from $K$ by subdividing the edge $uv$ by
a new vertex $w$. We also take $u_1=u, v_1=v$ and $u_2=u, v_2=w$. Then it is
easy to verify that $J_1, J_2$ satisfy the required properties. Indeed, in any acyclic
$r$-coloring of $J_1$, the vertices $u_1=u, v_1=v$ must obtain the same colour, 
otherwise $J_1 \cup uv = K$ would also be acyclically $r$-colorable. The same 
argument holds for $J_2$ and $u=u_2$ and $v$, and therefore $u_2=u$ and $v_2=w$
must obtain different colors. Any path between $u_1$ and $v_1$ in $J_1$ contains
at least $k$ vertices, otherwise $K$ would contain a cycle shorter than $k$.
\end{proof}

We are ready to tackle the desired result for the dichromatic number.

\begin{theorem}\label{2.2}
There exists a function $d=d(r,k)$ such that given $r,k\geq 3$,
the problem of acyclically $r$-coloring a digraph $G$ of directed girth at least  
$k$, and of in-degrees and out-degrees at most $d=d(r,k)$, is NP-complete. 
\end{theorem}

\begin{proof}
The proof is similar to the undirected case above. We again reduce from graph $r$-colorability. 
Let $G$ be an instance, and replace each vertex $x$ by an oriented binary tree $T_x$ with 
$deg(v)$ leaves. The tree is first rooted at a non-leaf vertex, then oriented away from the root.
We will use two digraph gadgets, $J_1$ with the following properties:

\begin{itemize}
\item
$J_1$ has vertices $u_1, v_1$ such that any acyclic $r$-coloring of $J_1$ assigns 
the same color to $u_1$ and $v_1$;
\item
$J_2$ has vertices $u_2, v_2$ such that any acyclic $r$-coloring of $J_2$ assigns 
different colors to $u_2$ and $v_2$;
\item
$J_1$ has directed girth at least $k$ and contains no directed path with fewer than $k$ vertices
from $u_1$ to $v_1$; and
\item
$J_2$ has directed girth at least $k$.
\end{itemize}

Each directed edge $st$ of each $T_x$ will be replaced by $J_1$ identifying $s$ with $u_1$ 
and $t$ with $v_1$, and each directed edge $xy$ of $G$ will similarly be replaced by a copy of 
$J_2$ between the corresponding leaves of $T_x$ and $T_y$. The resulting graph $G'$ has directed girth 
at least $k$ because every directed cycle in $G'$ is either inside a copy of some $J_i$, or passes 
through some $J_1$. Clearly, $G'$ is acyclically $r$-colorable if and only if $G$ is $r$-colorable 
in the usual sense. The in- and out-degrees of $G'$ are maximized by three times the maximum 
in- and out-degree of any $u_1, u_2, v_1, v_2$ in $J_1, J_2$.
\end{proof}

In this case there is again a simpler construction when $r=2$.

\begin{theorem}\label{2.1}
\label{few}
Given $k\geq 3$,
the problem of acyclic 2-coloring a digraph $G$ of directed girth at least
$k$ and of in-degrees and out-degrees at most $k+1$, is NP-complete. 
\end{theorem}

\begin{proof}
We proceed as in Theorem \ref{3.2}, reducing from 
the not-all-equal $k$-satisfiability problem with three
occurences of each variable $x$, replacing each clause
with a disjoint directed $k$-cycle. To ensure that the three
occurences of a variable $x$ in clauses have the same
value, we consider for each variable $x$ a separate digraph 
$H_k$ whose vertices are partitioned into $k$ sets $S_0,S_1,\ldots,S_{k-1}$ 
with $S_0$ of size one, $S_1$ an independent set of size three
and each $S_i$ for $2\leq i<k$ inducing a directed $k$-cycle. In addition, 
we include all edges from $S_i$ to $S_{i+1}$, and $S_{k-1}$ to $S_0$. 
In any acyclic two-coloring of $H_k$, each $S_i$ for $2\leq i<k$ must 
have both colors, so the colors in $S_1$ must all be different from the color
in $S_0$, and hence the same. The three elements of $S_1$ can thus 
be identified with the three occurences of $x$.
\end{proof}

\section{High Girth Graphs and Digraphs}

In this section we discuss the existence of high-girth graphs and digraphs without 
acyclic $r$-colorings. For ordinary graph $r$-colorings, we have the following 
well-known result of \cite{erdos}.

\begin{theorem}\label{PaliBacsi}
For any $r,k\geq 3$, there exists a graph with girth at least $k$ which is not $r$-colorable.
\end{theorem}

For dichromatic number we have the following theorem~\cite{bor2}.

\begin{theorem}\label{BojanBacsi}
For any $r,k\geq 3$, there exists a digraph with directed girth at least $k$ which is not 
acyclically $r$-colorable.
\end{theorem}

A very general version of such results is proved in \cite{fv} (Theorem 5). The proof in \cite{fv} is probabilistic
but there is a constructive proof in \cite{k}. We refer the reader to \cite{fv} for the definition of a constraint 
satisfaction problem, the girth of an instance, and equivalence of problems. We explain below the special 
case sufficient for our applications here.

\begin{theorem}\label{csps}
For every constraint satisfaction problem $P$, any instance $I$ of $P$, and any integer $k\geq 3$, 
there exists an instance $I'$, equivalent to $I$, with girth at least $k$.
\end{theorem}

The {\em $r$-valued not-all-equal $k$-satisfiability problem} is an example of a constraint satisfaction 
problem. Here an instance is a set of variables $x_1, x_2, \dots, x_n$ each taking one of $r$ possible
values, and a set of constraints $C_1, C_2, \dots, C_m$, each consisting of exactly $k$ variables. The
solution of an instance is an assignment of values to the variables such that no constraint $C_i$ has
all its variables assigned the same value. (This can also be viewed as the $r$-coloring problem of
$k$-uniform hypergraphs; cf. the special case $r=2$ in the proof of Theorem \ref{3.2}.) In this case,
the girth of an instance is the smallest set of variables $y_0, y_1, \dots, y_{k-1}$ such that any two
consecutive $y_j, y_{j+1}$ (subscripts modulo $k$) occur together in some constraint $C_i$. We say
that an instance $I$ is equivalent to an instance $I'$ if $I$ has a solution if and only if $I'$ has a
solution. There obviously are instances without a solution, for example $n=(k-1)r+1$ variables
and all $m={n \choose k}$ constraints imposed on each subset of size $k$. (If each variable is 
assigned one of $r$ values, some $k$ variables will have the same value, so $I$ has no solution.)
We obtain the following corollary of the theorem.

\begin{theorem}\label{tomasko}
For any $r,k\geq 3$, there exists an instance of $r$-valued not-all-equal $k$-satisfiability problem
with girth at least $k$, which has no solution.
\end{theorem}

We can transform the instance into a digraph by taking a vertex for each variable $x_i$ and form
a directed $k$-cycle on any set of $k$ variables occurring in a constraint $C_i$. Clearly, this digraph
has directed girth at least $k$. We obtain a new proof of Theorem \ref{BojanBacsi}.

By replacing each constraint with an undirected $k$-cycle, we similarly conclude the following useful fact.

\begin{theorem}\label{newest}
For any $r,k\geq 3$, there exists a graph with girth at least $k$ which is not acyclically $r$-colorable.
\end{theorem}

We close this section by noting that graph $r$-coloring is another example of a constraint satisfaction 
problem, and applying Theorem \ref{csps} to the graph $K_{r+1}$ which is not $r$-colorable, we
obtain Theorem \ref{PaliBacsi}.

The digraph $K$ from Theorem \ref{BojanBacsi} may be assumed to contain an edge $uv$ such that 
$K - uv$ is acyclically $r$-colorable, e.g., by assuming that $K$ is minimal with respect to inclusion.
A similar remark applies to the graph $K$ from Theorem \ref{newest}.

The obvious question is whether a more explicit construction for $H$ and $H'$
can be given, thus avoiding randomization~\cite{erdos,fv} or a more complex
construction~\cite{k}. For example, in the case $k=3$,
a random tournament as in Theorem~\ref{ran} of size
$O(r\log r)$ suffices for $H$, yet it remains hard to find $H$ and $H'$.
Our construction below gives $H=H'$ of size polynomial in $k$ for $r$ fixed,
or polynomial in $r$ for $k$ fixed.

\begin{theorem}\label{2.3}
For every $r\geq 1$, $k\geq 3$, there exists a digraph $H^k_r$ with the following properties.
\begin{enumerate}
\item
$H^k_r$ has at most $k^r$ vertices;
\item
moreover, if $k\leq r$, then $H^k_r$ has at most 
$$k^{\lceil k(1+\log(\frac{r}{k}))\rceil} \leq k{(\frac{er}{k})}^{k\log k}$$
vertices;
\item
$H^k_r$ can be constructed in time linear in the number of vertices;
\item
$H^k_r$ does not have an acyclic $r$-coloring;
\item
for each edge $uv$, the graph $H^k_r-uv$ does have an acyclic $r$-coloring; and 
\item
$H^k_r$ has directed girth $k$.
\end{enumerate}
\end{theorem}

This gives the bound $d(r,k)\leq 3|V(H^k_r)|$ in Theorem~\ref{3.1}.

\begin{proof}
We fix $k$, and let $H^k_1$ be a $k$-cycle, satisfying all conditions. 
For $H^k_r$ with $r\geq 2$, write
$r-1=a\lfloor\frac{r-1}{k}\rfloor + b\lceil\frac{r-1}{k}\rceil$ with
$a,b\geq 0$ and $a+b=k$.

Let $r'=r-1-\lfloor\frac{r-1}{k}\rfloor$
and $r''=r-1-\lceil\frac{r-1}{k}\rceil$. Define $H^k_r$ as the disjoint union
of $a$ copies of $H^k_{r'}$ and $b$ copies of $H^k_{r''}$, for a total of $k$
copies, with all edges joining each such copy to the next,
or the last one to first one.

We prove the last three conditions by induction on $r$.
The first $a$ copies need at least $r'+1$ colors, avoiding at most only
$\lfloor\frac{r-1}{k}\rfloor$ colors.
The last $b$ copies need at least $r''+1$ colors, avoiding at most only
$\lceil\frac{r-1}{k}\rceil$ colors.
By the definition of $a,b,$ at most $r-1$ colors are avoided by at least one
copy, so some color $i$ appears in all the copies, and this gives a cycle of
color $i$ of length $k$ across all the copies. This proves condition (4).

For condition (5), suppose the removed edge $uv$
is inside the $j$th copy $H_j$. Then $H_j$ can be colored with only $r'$
colors ($r''$ colors), giving one more color than in the definition of $a,b,$
for a total of $r$ avoided colors across all the copies, so no color appears
in all the copies, and there is no cycle across all the copies that gives the
same color in all copies. This proves condition (5) in this case.

If the removed edge is $uv$ joins say $H_j,H_{j+1},$ then color $H_j-u$ and
$H_{j+1}-v$ with only $r'$ (or $r''$) colors by condition (5), avoiding one more 
color in each of $H_j,H_{j+1}$, with only color $i$ avoided in both cases. 
Assign color $i$ to $u,v$. This gives us again $r-1$ avoided colors, but
including color $i$ in all copies does not gives a cycle of length $k$ of
color $i$, since the cycle would have to go through edge $uv$. This proves
condition (5).

For condition (6), note that all cycles either go through only one
copy and are thus inductively of length at least $k$, or go through all the
copies and must thus be of length at least $k$.

For conditions (1, 2), note that each step of the induction
has $r',r''\leq r(1-\frac{1}{\min(r,k)})$ and $k|V(H^k_{r'})|\geq |V(H^k_r)|$.
\end{proof}

Note that this last result allows us to prove Theorems~\ref{few} and~\ref{3.1}
without necessarily assuming that $r,k$ are constants, but may depend on $n$,
for as long as the bound $|V(H^k_r)|\leq n^{1-\epsilon}$ holds with
$\epsilon>0$ constant.

\section{Approximation}

We now prove a hardness of approximation result.

\begin{lemma}
Let $0<\epsilon<1$ be a constant. Given a complete bipartite graph
$H=(U,V,E)$ with $|U|=|V|=n$,
let $H'$ be the digraph obtained from $H$ by orienting the edges in either
direction independently with equal probability $\frac{1}{2}$. Then with
probability approaching 1 as $n$ goes to infinity, for every two subsets
$U'\subseteq U, V'\subseteq V$ having $|U'|,|V'|\geq n^{\epsilon}$,
the digraph induced by $U'\cup V'$ contains a cycle.
\end{lemma}

\begin{proof}
Say $|U'|=|V'|=n^{\epsilon}$ and choose $U',V'$ ordered in at most
$n^{2n^{\epsilon}}$ ways. If
$U'\cup V'$ induces an acyclic subgraph consistent
with these ordering, then the order of the neighbors of a vertex $u\in U'$,
for some such order, will have first the edges incoming to $u$ from $V'$,
then the outgoing edges from $u$ to $V'$. This can happen in $n^{\epsilon}+1$
ways out of $2^{n^{\epsilon}}$ possible choices for the edes joining $u$.
Multiplying resulting probability over all $n^{\epsilon}$ choices of $u$
from choices of subsets $U',V'$ gives probability of there being $U'\cup V'$
acyclic at most
$$n^{2n^{\epsilon}}{(n^{\epsilon}+1)}^{n^{\epsilon}}2^{-n^{2\epsilon}},$$
which tends to zero as $n$ goes to infinity.
\end{proof}

Noga Alon informed us that the known relatively recent explicit construction for 
bipartite Ramsey graphs~\cite{bar} will give a derandomization of this lemma.
Indeed, if $U'\cup V'$ induces an acyclic digraph with a corresponding linear
order $L$, then the middle vertices of $U',V'$ in $L$ are $u,v,$ respectively.
Say the edge joining $u,v$ goes in the direction $uv$. Then there are edges
going from vertices below $u$ in $U'$ to vertices above $v$ in $V'$. (The other
case is symmetric, from below $v$ in $V'$ to above $u$ in $U'$ if the direction
is $vu$). We can define a biparite graph from $H'$ by including only edges
from $U$ to $V$. Then we just saw that we would have either a bipartite clique
or a biparite independent set with $k$ vertices in each side,
$k=n^{\epsilon}/2$. The bipartite Ramsey construction in~\cite{bar} 
guarantees this does not happen even with $k=n^{o(1)}$.

Feige and Kilian~\cite{fk} proved that it is NP-hard to find an independent
set of size greater than $n^{\epsilon}$ (thus hard to $n^{1-\epsilon}$ color)
in a graph $G$ that is colorable with
$n^{\epsilon}$ colors, for any $\epsilon>0$. As a result, it is equally hard
to find a large acyclic subgraph in a digraph, since we could replace the
edges of $G$ with digons (girth 2), so that acyclic sugraphs in the resulting
digraph correspond to independent sets in $G$.

\begin{theorem}\label{2.5}
Fix $0<\epsilon<\frac{1}{4}$. It is NP-hard
to find an acyclic induced subgraph of size greater that
$N^{\frac{1}{2}+\epsilon}$ (thus hard to find an acyclic
$N^{\frac{1}{2}-\epsilon}$ coloring)
in an $N$-vertex digraph $G'$ without digons, i.e., of girth at least 3,
that has an acyclic $r$-coloring with $r\leq N^{\epsilon}$.
\end{theorem}

\begin{proof}
Let $G$ be an instance of the NP-hard question of Feige and Killian.
For each vertex $v_i\in V(G)$, let $U_i$ be a set with $|U_i|=n$.
For each edge $v_iv_j\in E(G)$, join $U_i,U_j$ with the random bipartite
digraph as in the lemma (which can be derandomized).
This gives a digraph $G'$ with $N=n^2$ vertices that
has an acyclic $r$-coloring with $r\leq N^{\frac{\epsilon}{2}}$,
by copying each color of a $v_i$ into the corresponding $U_i$.

However, an acyclic induced subgraph $S$ can only meet sets $U_i$ in at least 
$n^{\epsilon}$ vertices if the corresponding $v_i$ form an independent set, by 
the lemma. Therefore there will only be found $n^{\epsilon}$ such large intersections
by the result of Feige and Killian, giving us $n^{1+\epsilon}$ vertices of $S$, plus 
small intersections, of size at most
$n^{\epsilon}$ for the remaining $U_i$, for a total 
$|S|\leq 2n^{1+\epsilon}=2N^{\frac{1+\epsilon}{2}}.$
\end{proof}

\section{Random Tournaments}

We begin with two simple observations to introduce random tournaments,
as in~\cite{erd2}.

\begin{theorem}
\label{short}
Every tournament $T$ on $n$ vertices contains a transitive subtournament
on $\lceil \log_2 (n+1)\rceil)$ vertices, and therefore has an acyclic
$\frac{n}{(1-\epsilon)\log_2 n}+n^{1-\epsilon}=O(\frac{n}{\log n})$ coloring.
\end{theorem}

\begin{proof}
Greedily select a vertex $v$ from $T$ of outdegree at least $\frac{n-1}{2}$,
remove $v$ and its in neighbors from $T$ to obtain $T'$ of size at least
$\frac{n-1}{2}$. This halving can be done $\lceil\log_2 (n+1)\rceil)-1$ times.
The chosen $\lceil \log_2 (n+1)\rceil)$
vertices $v$ will form a transitive tournament.
For the acyclic coloring, select and remove greedily transitive tournaments
from $T$, each of size at least $(\log_2 n){1-\epsilon}$,
until we reach a tournament $T'$ of size at most $n^{1-\epsilon}$,
and we use at most these many colors for $T'$.
\end{proof}

Random tournaments essentially match the preceding bound

\begin{theorem}
\label{ran}
With high probability,
a tournament $T$ on $n$ vertices only contains a transitive subtournament
on $O(\log n)$ vertices, and therefore only has an an acyclic
$\Omega(\frac{n}{\log n})$-coloring.
\end{theorem}

\begin{proof}
Selecting a sequence of $\lfloor 3\log n\rfloor$ distinct vertices $v$ of $T$
can be done in at most $s=n^{3\log n}$ ways. The probability that such a
sequence will give the ordering of a transitive tournament is at most
$\frac{1}{t}$ for $t=2^{{3\log n - 1}\choose{2}}$. The ratio
$\frac{s}{t}$ tends to 0 as $n$ goes to infinity.
\end{proof}

We now define a random model for acyclic $r$-colorable tournaments $T$ on $n$
vertices. Consider $r$ integers $s_1\geq s_2\geq \cdots\geq s_r\geq 1$ with
$s_1+s_2+\cdots + s_r=n$. To define $T$, first consider $r$ disjoint sets of
vertices $S_i$ with $|S_i|=s_i$, and impose on each $S_i$ a transitive
(acyclic) tournament. Finally, orient each edge joining vertices in two
different $S_i$ independently with probability $\frac{1}{2}$ in either
direction.

The tournaments $T$ so generated have an acyclic $r$-coloring, obtained by
assigning color $i$ to the vertices in $S_i$. We consider the problem of
acyclic $r$-coloring such a $T$ when the vertices of $T$ are given in arbitrary
order, and the $S_i$ and $s_i$
are not given. We give a deterministic algorithm that
finds such a coloring with probability arbitrarily close to 1. If $r$ is fixed,
the algorithm runs in time $O(n^2)$, linear in the size of the input $T$.

The algorithm runs in three phases, which we describe below.

The first phase starts with the tournament $T_{n_1,r_1}=T$ with $n_1=n$ and
$r_1=r$, and operates in rounds. At the beginning of the $j$th round, we have
$T_{n_j,r_j}$.

We define $d_j=c\sqrt{n_j}\log n_j$ for some constant $c$. Given a tournament
$R$ and a vertex $v$ in $R$, we define
$$d^R_{\rm diff}(v)=d^R_{\rm out}(v)-d^R_{\rm in}(v)$$
as the difference between the out-degree and the in-degree of $v$ in $R$.
Note that $E(d^{T_{n_j,r_j}}_{\rm diff}(v))=d^{S_i}_{\rm diff}(v)$ if
$v\in S_i$.
Consider the Chernoff bounds for $X$ equal to
the sum of independent Bernoulli random
veriables, with $\mu=E(X)$.
$$Pr(X\leq (1-\delta)\mu)\leq e^{-\frac{\delta^2\mu}{2}}, 0\leq\delta\leq 1,$$
$$Pr(X\geq (1+\delta)\mu)\leq e^{-\frac{\delta^2\mu}{3}}, 0\leq\delta\leq 1.$$
Letting
$$X=d^{T_{n_j,r_j}}_{\rm out}(v)-d^{S_i}_{\rm out}(v))$$
we have that
$$Pr(|X-\frac{n_j-s_i}{2}|\geq d_j)\leq 2e^{-\frac{2d_j^2}{3(n_j-s_i)}}$$

$$\leq 2e^{-\frac{2}{3}c^2\log^2 n_j}$$
and therefore
$$Pr(|d^{T_{n_j,r_j}}_{\rm diff}(v)-d^{S_i}_{\rm diff}(v)|\geq 2d_j)
\leq 2e^{-\frac{2}{3}c^2\log^2 n_j}.$$
The probability that this holds for all $v$ in $T_{n_j,r_j}$ is at most $n_j$
times the bound. Let $u^*$ be the vertex that maximizes
$|d^{T_{n_j,r_j}}_{\rm diff}(v)|$, say the quantity within the absolute value
is nonnegative. Then $u^*\in S^*=S_{i^*}$ with $s^*=|S^*|$.
Let $S_1$ be the largest of the
$S_i$ in $T_{n_j,r_j}$, and let $u_1$ be the starting vertex of the
transitive tournament $S_1$. Then
$$|S_1|-|S^*|\leq d^{S_1}_{\rm diff}(u_1)-d^{S^*}_{\rm diff}(u^*)$$
$$\leq (d^{S_1}_{\rm diff}(u_1)-d^{T_{n_j,r_j}}_{\rm diff}(u_1))
+(d^{T_{n_j,r_j}}_{\rm diff}(u_1)-d^{T_{n_j,r_j}}_{\rm diff}(u^*))
+(d^{T_{n_j,r_j}}_{\rm diff}(u^*)-d^{S^*}_{\rm diff}(u^*))$$
$$\leq 2d_j+0+2d_j=4d_j$$
with probability at least $1-2(n_j+1)e^{-\frac{2}{3}c^2\log^2 n_j}$.

The algorithm seeks to determine $S^*$ given the vertex $u^*\in S^*$.
For $v\neq u^*$, and $w\neq u^*,v$, the probability
that $u^*vw$ is a directed 3-cycle in either direction is $\frac{1}{4}$,
unless $v,w\in S^*$, in which case the probability is zero.
Let $X(v)$ be the number of $w$s that form such a directed 3-cycle
with $u^*$ and $v$. Then $E(X(v))=\frac{1}{4}(n-2)$ if $v\notin S^*$,
and $E(X(v))=\frac{1}{4}(n-s^*)$ if $v\in S^*$. Then
$$Pr(|X(v)-E(X(v))|\geq d_j)\leq 2n_je^{-\frac{4d_j^2}{3n_j}}$$
$$\leq 2n_je^{-\frac{4}{3}c^2\log^2 n_j},$$
where the factor of $n_j$ is needed to account for the fact that $u^*$ could
be any vertex in $T_{n_j,r_j}$.
With probability $n_j$ times this much this holds for all $X(v)$.

Suppose $s^*-2>2d_j$.
With probability at least $1-4n_j^2e^{-\frac{4}{3}c^2\log^2 n_j},$
vertices $v\in S^*$ have $X(v)\leq A'=\frac{1}{4}(n-s^*)+d_j$, and
vertices $v\notin S^*$ have $X(v)\geq B'= \frac{1}{4}(n-2)-d_j$
and $B'>A'$.

Thus for as long as $s_1>6 d_j+2$, we have $s^*>2 d_j+2$, and
and the algorithm can determine $u^*$ and test the $A',B'$ bounds to determine
$S^*$. This works with probability at least
$$1-p_j=1-2(n_j+1)e^{-\frac{2}{3}c^2\log^2 n_j}
-4n_j^2e^{-\frac{4}{3}c^2\log^2 n_j}$$
$$=1-e^{-(1-\epsilon_j)\frac{2}{3}(c^2\log^2 n_j)}$$
where $\epsilon_j\leq \frac{3}{c^2\log n_j}$.

\begin{lemma}
Suppose the $j$th round starts with $T_{n_j,r_j}$.
Let  $d_j=c\sqrt{n_j}\log n_j$.
With probability at least $1-p_j$,
if the largest $S_1$ has $s_1>6d_j+2$, then we find $S^*$ with
$s^*\geq s_1-4d_j>2d_j+2$, and reduce the problem to
$T_{n_{j+1},r_{j+1}}=T_{n_j-s^*,r_j-1}$.
\end{lemma}

The last round of the first phase takes
$T_{n_j,r_j}$ to $T_{n_{j+1},r_{j+1}}$.
We let $n^-=n_j, r^-=r_j, n'=n_{j+1}, r'=r_{j+1},$ and $\hat{j}=j$ for this
$j$. Note that after phase one is over, we have
$s'_1\leq 6d^+2\leq 2+6c\sqrt{n'}\log n'$ with $s'_1$ and $d'$ defined
similarly with $j=\hat{j}$.
When phase one no longer applies,
$r'\geq\frac{n'}{s'_1}\geq \frac{\sqrt{n'}}{7c \log n'}$.

In particular, if $r=O(1)$, then $r'=O(1)$ and $n'=O(1)$
and the rest of the problem can be solved in $O(1)$
time.

We may assume $n^->\max(\frac{\log n}{\log r},\frac{\log n}{\log\log n})$
since otherwise the problem can be solved in linear time 
avoiding the $\hat{j}$th round.

\begin{theorem}
Finding $S^*$ takes $O(n^2)$ time, linear in the size of the input $T$. 
This yields a total running time of $O(rn^2)$ over at most $r$ rounds.
The probability bound for the first phase is $1-n_j p_j=1-e^{-\Omega(c^2\log^2 n^-)}$
with $j=\hat{j}$.

For $r$ constant, we can find an acyclic $r$-coloring of $T$
on $n$ vertices in time $O(n^2)$, linear in the size of the input,
with probability as above.
\end{theorem}

After the first phase is over, we have $T'=T_{n',r'}$.
The second phase first identifies all sets $U$ in $T'$
with $|U|\leq c\log n'$ that could be the least elements 
of an $S_i$. The number of possible such sets $U$ is 
at most ${n'}^{c\log n'}$.

First $U$ must be a transitive tournament. Suppose
$U$ of size $c\log n'$
is the bottom of an $S_i$. Let $V$ be
$U$ plus all the elements above all of $U$.

We claim that $|V|\setminus S_i|<c\log n'$. Otherwise
choose $W$ with $|W|=c\log n'$ contained in $V\setminus S_i$.
The sets $U, W$ are joined by edges joining different
colors, thus the probability that they will all be oriented
upwards is at most $2^{-c^2\log^2 n'}$. There are at most
${n'}^{2c\log n'}$ possible choices of $V_{\ell},W$, so
with probability at least $1-e^{(2c-c^2\log 2)\log^2 n'}$
the claim follows.

Suppose $k_0\geq c_0\log_2 n'$ for ssome sufficiently large constant $c_0$.
The algorithm repeatedly chooses sets $Z=V\setminus W$ that
define transitive (acyclic) tournaments within $T_{n',r'}$.
The sets $Z$ are considered in nonincresasing order of $z=|Z|\geq k_0$.
We claim that with high probability,
we will have $Z=S_i$ for one of the sets $S_i$,
so we choose such a $Z$ and discard all later $Z'$ that intersect $Z$,
with $|Z'|\leq |Z|$. The algorithm ends when there are no more $Z$
with $z\geq k_0$.

\begin{theorem}
The second phase correctly selects the remaining $S_i$ with $s_i\geq k_0$,
with probability at least $1-e^{(2c-c^2\log 2)\log^2 n'}-2^{-\frac{k_0}{8}+1}$.
for $k_0\geq\geq 24\log n'$ and $n'$ sufficiently large.
The running time is bounded by $O(e^{3c\log^2 n'})$.
\end{theorem}

\begin{proof}
It only remains to show that all chosen $Z$ are $S_i$, with probability
at least $1-2^{-\frac{k_0}{8}+1}$. If not, $Z$ meets at least two $S_i$. Let
$Z=\bigcup_i Z_i$, where $Z_i=Z\cap S_i$, with $z_i=|Z_i|$.
Order the $Z_i$ in decresasing order of $z_i$, and let $\hat{i}$ be such that
$\sum_{i<\hat{i}} z_i<\frac{z}{2}$ and $\sum_{i>\hat{i}} z_i<\frac{z}{2}$.

Let $\ell=z-z_{\hat{i}}$.
The $z_i$ can be partitioned into two sets
into one of the two cases
$A=\{z_1,\ldots,z_{\hat{i}}\},B=\{z_{\hat{i}+1},\ldots,z_t\}$ or
$A=\{z_1,\ldots,z_{\hat{i}-1}\},B=\{z_{\hat{i}},\ldots,z_t\}$,
and one of these two
partitions has corresponding sizes at least $\frac{z}{2},\frac{\ell}{2}$.

Once the edges within $A$ with $|A|\geq\frac{z}{2}$
have been oriented, the probability that a vertex
$w$ in $B$ with $|B|\geq\frac{\ell}{2}$ will fit in some order among $A$ is
$(|A|+1)2^{-|A|}=(\frac{z}{2})2^{-\frac{z}{2}},$ or
${((\frac{z}{2})2^{-(\frac{z}{2})})}^{\frac{\ell}{2}}$ over $B$.
The $\ell$ vertices in $Z-Z_1$ and at most $\ell$ vertices in $S_1-Z_1$
(since $s_1\leq z$) can be chosen in at most $n^{2\ell}$ ways, giving the
probability bound
$$n^{2\ell}{((\frac{z}{2})2^{-(\frac{z}{2})})}^{\frac{\ell}{2}}$$
$$= 2^{-\ell{(-2\log_2 n -\frac{\log_2 \frac{z}{2}}{2}+\frac{z}{4})}} $$
$$= 2^{-\frac{\ell z}{4}
(1-(\log_2 e)(8\frac{\log n}{z}+\frac{2\log\frac{z}{2}}{z}))}$$
$$\leq 2^{-\frac{\ell z}{8}}$$
$$\leq 2^{-\frac{\ell k_0}{8}}$$
for $z\geq k_0\geq 24\log n'$ and $n'$ sufficiently large.
Summing over all $\ell\geq 1$ gives the bound
$\leq 2^{-\frac{k_0}{8}+1}$.
\end{proof}

The third phase starts with a resultant $T_{n'',r''}$ and possible
sets $Z$ with $z<24\log n'$, so $r''\geq\frac{n''}{24\log n'}$.
Each $S_i$ has $e^{2c\log^2 n'}$ choices of $U,W$, for a total of
$e^{2cr''\log^2 n'}$ choices for the $r''$ sets $S_i$ to be selected,
giving running time $O(n''e^{2cr''\log^2 n'})$.

The third phase seems the most expensive, since the running time
is exponential in $r$ versus quasi-polynomial in the second phase,
and polynomial in the first phase. We can avoid the third phase by
approximating the bound $24\log n'$ on color classes with the bound
$\log_2 n$ from Theorem~\ref{short}, for an approximation factor
of $24\log 2$.

\end{document}